\documentclass{article}
\usepackage[dvips,letterpaper,margin=1.0in]{geometry}
\usepackage{sty,thm-restate}

\title{Is Planted Coloring Easier than Planted Clique?}

\date{}

\usepackage{authblk}
\author[1]{Pravesh K.\ Kothari\thanks{Email: \textit{praveshk@cs.cmu.edu}. Supported by  NSF CAREER Award \#2047933, Alfred P.\ Sloan Fellowship and a Google Research Scholar Award.}}
\author[2]{Santosh S.\ Vempala\thanks{Email: \textit{vempala@gatech.edu}. Supported in part by NSF awards CCF-2007443 and CCF-2106444.}}
\author[3]{Alexander S.\ Wein\thanks{Email: \textit{aswein@ucdavis.edu}. Part of this work was done while with the Algorithms and Randomness Center at Georgia Tech, supported by NSF awards CCF-2007443 and CCF-2106444.}}
\author[1]{Jeff Xu\thanks{Email: \textit{jeffxusichao@cmu.edu}. Supported in part by NSF CAREER Award \#2047933.}}

\affil[1]{Computer Science Department, Carnegie Mellon University}

\affil[2]{School of Computer Science, Georgia Tech}

\affil[3]{Department of Mathematics, University of California, Davis}

\begin{document}

\maketitle

\begin{abstract}
We study the computational complexity of two related problems: recovering a planted $q$-coloring in $G(n,1/2)$, and finding efficiently verifiable witnesses of non-$q$-colorability (a.k.a.\ refutations) in $G(n,1/2)$. Our main results show hardness for both these problems in a restricted-but-powerful class of algorithms based on computing low-degree polynomials in the inputs.

The problem of recovering a planted $q$-coloring is equivalent to recovering $q$ disjoint planted cliques that cover all the vertices --- a potentially easier variant of the well-studied planted clique problem. Our first result shows that this variant is as hard as the original planted clique problem in the low-degree polynomial model of computation: each clique needs to have size $k \gg \sqrt{n}$ for efficient recovery to be possible. For the related variant where the cliques cover a $(1-\epsilon)$-fraction of the vertices, we also show hardness by reduction from planted clique.

Our second result shows that refuting $q$-colorability of $G(n,1/2)$ is hard in the low-degree polynomial model when $q \gg n^{2/3}$ but easy when $q \lesssim n^{1/2}$, and we leave closing this gap for future work. Our proof is more subtle than similar results for planted clique and involves constructing a non-standard distribution over $q$-colorable graphs. We note that while related to several prior works, this is the first work that explicitly formulates refutation problems in the low-degree polynomial model.

The proofs of our main results involve showing low-degree hardness of hypothesis testing between an appropriately constructed pair of distributions. For refutation, we show \emph{completeness} of this approach: in the low-degree model, the refutation task is precisely as hard as the hardest associated testing problem, i.e., proving hardness of refutation amounts to finding a ``hard" distribution.
\end{abstract}

\section{Introduction}

The \emph{planted clique} problem, introduced by \cite{Jerrum} and \cite{kucera}, asks for a polynomial-time algorithm to find a clique of size $k$ added to an \ER random graph $G(n,1/2)$. The associated task of \emph{refuting} the existence of $k$-cliques in $G \sim G(n,1/2)$ asks for a polynomial-time algorithm to compute a certificate that can be efficiently verified to infer the absence of a $k$-clique in $G$. Despite a long line of work, state-of-the-art polynomial-time algorithms for both problems~\cite{AKS98} only succeed when $k = \Omega(\sqrt{n})$. In contrast, the clique number of $G(n,1/2)$ is at most $ \lceil 2 \log_2 n \rceil+1$ with high probability and thus, an added clique of any size $k > \lceil 2 \log_2 n \rceil+1 $ is uniquely identifiable. A long line of work proving lower bounds in various restricted models such as Markov chains~\cite{Jerrum}, the Statistical Query model~\cite{sq-clique}, convex relaxations~\cite{FeigeK03} and in particular the sum-of-squares hierarchy and the related low-degree polynomial model of computation~\cite{BHKK+16,hopkins-thesis}, suggest that the \emph{algorithmic threshold} for both variants --- the smallest $k$ for which efficient algorithms can find the added $k$-clique or refute the existence of $k$-cliques in $G(n,1/2)$ --- is $\Omega(\sqrt{n})$. In the past two decades, the hypothesis that no polynomial-time procedure can beat the above guarantees of the known algorithms has become a focal point in average-case complexity theory and the root of myriad reductions to average-case problems arising in various domains (e.g.,~\cite{BR-reduction,HWX-reductions,BBH-reductions,hardness-of-Nash-welfare}).

\paragraph{Two motivating problems: recovery and refutation of $q$-colorings.}

In this paper, we study the following innocuous-looking (and ostensibly easier than planted clique) question where, in the \emph{recovery} problem, we study the complexity of exactly recovering $\approx n/k$ disjoint planted $k$-cliques in $G \sim G(n,1/2)$, with high success probability. If the disjoint planted cliques cover all the vertices of the graph, then the complement of the graph has a planted $(n/k)$-coloring. Thus, this version of our problem is tantamount to studying whether recovering a {\em planted $q$-coloring} in $G(n,1/2)$ is easier than recovering a single planted clique. In the associated \emph{refutation} problem, the goal is to find an algorithm that takes as input a graph $G$ and outputs NO or MAYBE with the guarantee that (1) whenever it outputs NO, the graph must not admit a valid $q$-coloring of its vertices, and (2) when $G \sim G(n,1/2)$, the algorithm should output NO with probability $1-o(1)$ over the draw of $G$.

The relation between the recovery and refutation tasks is somewhat subtle: while these two problems appear related, we are not aware of a formal reduction between them in either direction. In this paper, we study the recovery and refutation problems separately, and draw attention to the fact that rather different methods will be needed to prove lower bounds in the two settings. We note that for colorability of \emph{sparse} random regular graphs, there appears to be a constant-factor gap between the recovery and refutation thresholds~\cite{spectral-planting}.

\paragraph{Proof strategy: hypothesis testing.}

One common strategy to understand the complexity of recovery or refutation is to introduce an auxiliary \emph{hypothesis testing} task: given a graph $G$ that is sampled either from some ``null'' distribution $\QQ$ (e.g., $G(n,1/2)$) or some ``planted'' distribution $\PP$ (e.g., some distribution supported on $q$-colorable graphs), design an efficiently computable statistical test that decides which of the two distributions generated a given sample $G$, with high success probability over the draw of $G$. Note that if there is an efficient refutation algorithm for some distribution $\QQ$, then we immediately obtain an efficient distinguisher between $\QQ$ and \emph{any} distribution $\PP$ supported on $q$-colorable graphs. Similarly, if there is an efficient recovery algorithm for some distribution $\PP$, then we immediately obtain a distinguishing algorithm between $\PP$ and \emph{any} distribution $\QQ$ supported on non-$q$-colorable graphs. As a result of this connection, we can conclude: 
\begin{itemize}
    \item [(I)] To show computational hardness of exactly recovering a planted $q$-coloring in a particular planted distribution $\PP$, it suffices to construct a null distribution $\QQ$ such that (i) with high probability, $G \sim \QQ$ is not $q$-colorable and (ii) it is computationally hard to distinguish $\PP$ from $\QQ$.
    \item [(II)] To show computational hardness of refuting $q$-colorability for a particular null distribution $\QQ$, it suffices to construct a planted distribution $\PP$ such that (i) $\PP$ is supported on $q$-colorable graphs and (ii) it is computationally hard to distinguish $\PP$ from $\QQ$.
\end{itemize}
Note that we have flexibility to choose either $\QQ$ (if studying recovery) or $\PP$ (if studying refutation). We will see later that it can be a non-trivial task to construct the right distribution. It need not be the case that the same testing problem arises when studying recovery as when studying refutation.

Strategy (II) has been referred to as constructing a \emph{computationally quiet planted distribution}~\cite{sk-cert}, where ``quiet'' pertains to the fact that the planted structure's presence cannot be detected by an efficient algorithm. Similarly, strategy (I) corresponds to constructing a \emph{computationally quiet null distribution}.

Since proving lower bounds for average-case hypothesis testing problems based on standard hardness assumptions is an elusively difficult goal at present (notwithstanding the recent successes~\cite{BBH-reductions,secret-leakage} that use the hardness of planted clique and its variants as a starting point in certain limited settings), we will obtain evidence of hardness for testing problems by focusing on a restricted but powerful and well-studied family of tests that we next describe.

\paragraph{Low-degree testing.} The low-degree polynomial model of hypothesis testing restricts the class of tests to be  polynomial functions in a natural representation of the input, with the complexity of a test captured by the degree of the polynomial. Specifically, viewing graphs as elements of $\{-1,1\}^{{n \choose 2}}$ with a $\{\pm 1\}$-indicator of presence or absence of every possible edge, the low-degree polynomial tests informally correspond to computing thresholds of arbitrary degree-$D$ polynomials of the edge-indicator variables. Since degree-$D$ polynomials can be computed (when described in the monomial coefficient representation) in time $n^{O(d)}$, constant-degree tests yield polynomial-time distinguishing algorithms. Despite being restricted, these low-degree tests already capture tests based on basic statistics of graphs such as edge counts, triangle counts, and more generally small subgraph counts (the number of edges in the subgraph corresponds to the degree of the polynomial). Various spectral methods (e.g., the leading eigenvalue of the adjacency matrix, or some other symmetric matrix whose entries are low-degree polynomials of the input variables) can also be approximated by polynomial tests of logarithmic degree in the number of variables; see~\cite[Section~4.2.3]{ld-notes}. As a result, low-degree tests (with degree $O(\log n)$) already capture the best known polynomial-time algorithms for a wide variety of high-dimensional statistical testing tasks (although we won't attempt to precisely characterize which tasks here; see e.g.~\cite{sos-detecting,hopkins-thesis,ld-notes,ld-counterexamples,lll} for discussion). As a result, if we manage to establish that all degree-$D$ tests provably fail to solve a particular testing problem for some $D = \omega(\log n)$, we say the problem is ``low-degree hard.'' This can be viewed as evidence suggesting computational hardness of the hypothesis testing problem. This is a widely-applicable and by now, commonly-used framework that originated in a line of work on proving lower bounds against the sum-of-squares hierarchy~\cite{BHKK+16,HS-bayesian,sos-detecting} (see also~\cite{hopkins-thesis,ld-notes} for further exposition).

\paragraph{Summary of results.}

Our main results use strategies (I) and (II) described above to shed light on the computational complexity of recovery and refutation of $q$-coloring. The formal models and statements are presented in the next section, but here we give a brief overview. Throughout, we will implicitly assume an asymptotic regime $n \to \infty$ where other parameters (e.g., $q, k$) may scale with $n$. We say an event occurs ``with high probability (w.h.p.)'' if it has probability $1-o(1)$ as $n \to \infty$. Since our focus is on identifying computational thresholds up to the correct power of $n$, we use the symbol $\ll$ in our informal discussions to hide factors of $n^{o(1)}$.

Our main result for the recovery problem shows that adding $\approx n/k$ disjoint cliques of size $k$ (instead of a single one) does not make the problem of recovering the added planted cliques easier. That is, our lower bounds suggest that each added clique needs to be of size  $\gtrsim \sqrt{n}$ for efficient recovery to be possible. 

In contrast and perhaps surprisingly, it turns out that adding more cliques makes the problem of \emph{distinguishing} the planted graph from $G(n,1/2)$ easier, simply by counting the total number of edges. This reveals a \emph{detection-recovery gap}, in contrast to the single planted clique problem (see Section~\ref{sec:reduction}).

More precisely, our results for recovery are as follows:
\begin{itemize}
    \item In the planted partial-coloring model where some fraction of the vertices are colored (equivalently, many disjoint planted cliques in $G(n,1/2)$ that cover at most a $(1-\epsilon)$-fraction of the graph), we show that:
    \begin{itemize}
    \item[(i)] If each clique has size $k \gg \sqrt{n}$, a simple algorithm can be used to recover them.
    \item[(ii)] If each clique has size $k \ll \sqrt{n}$, it is computationally hard to recover them assuming the Planted Clique Hypothesis. That is, recovering many planted $k$-cliques is as hard as recovering a single planted $k$-clique.
    \end{itemize}
    \item In the full planted coloring model ($q$ planted cliques of size $k = n/q$ partitioning the entire graph), we are unable to show hardness via reduction, but instead give an indirect argument that supports the same conclusion as above:
    \begin{itemize}
    \item[(i)] If each clique has size $k \gg \sqrt{n}$, there is again a simple algorithm to recover them.
    \item[(ii)] If each clique has size $k \ll \sqrt{n}$, we argue that recovery is computationally hard via strategy (I), taking the null distribution $\QQ$ to be a planted $(q+1)$-coloring. In other words, we prove that low-degree tests cannot even distinguish a planted $q$-colorable graph from a planted $(q+1)$-colorable graph. This suggests hardness of recovery via a two-stage argument described in Section~\ref{sec:q-vs-ql}.
    \end{itemize}
\end{itemize}

 For the problem of refuting $q$-colorability in $G(n,1/2)$, it is known that a poly-time algorithm exists when $k \coloneqq n/q \gg \sqrt{n}$ \cite{lovasztheta}. To explore the complexity of this problem, we explicitly formulate the refutation problem in the low-degree polynomial model (for the first time), and show the following:

\begin{itemize}
    \item If $k \gtrsim \sqrt{n}$ (i.e., $q \lesssim \sqrt{n}$), then there is a low-degree polynomial that refutes $q$-colorability in $G(n,1/2)$.
    \item If $k \ll n^{1/3}$ (i.e., $q \gg n^{2/3}$), then no low-degree polynomial refutes $q$-colorability in $G(n,1/2)$. The proof follows strategy (II) and involves constructing a non-trivial planted distribution $\PP$.
    \item We conjecture $k \sim \sqrt{n}$ is the true low-degree refutation threshold, and we leave this to future work. One way to improve the lower bound would be to construct a ``quieter'' planted distribution, i.e., a distribution supported on $q$-colorable graphs that is low-degree hard to distinguish from $G(n,1/2)$ whenever $k \ll \sqrt{n}$. Our final result is a duality argument showing that in fact, the conjecture is \emph{equivalent} to the existence of such a planted distribution.
\end{itemize}

\section{Results}

A central concept in this work will be that of {\em hypothesis testing} between two high-dimensional distributions. We consider two (sequences of) distributions $\PP = \PP_n$ and $\QQ = \QQ_n$. For us, these distributions will always be over $n$-vertex graphs. We use the following asymptotic notion of successful testing.

\begin{definition}[Strong distinguishing]\label{def:strong-det}
For two distributions $\PP_n$ and $\QQ_n$, we say an algorithm $A_n$ \emph{strongly distinguishes} $\PP$ and $\QQ$ if it takes as input a sample drawn from one of the two distributions and correctly determines which distribution it came from with probability $1-o(1)$ as $n \to \infty$. In other words, both type I and type II error probabilities must be $o(1)$.
\end{definition}

We will also be interested in the following class of ``low-degree'' tests. A degree-$D$ test is simply a (multivariate) polynomial in the input variables (or rather a sequence of such polynomials, one for each problem size $n$). In our case, there will be $\binom{n}{2}$ input variables --- one for every possible edge in an $n$-vertex graph --- taking values in $\{\pm 1\}$, where $+1$ indicates the presence of an edge and $-1$ indicates the absence. We use the following standard notion of ``success'' for a polynomial test.

\begin{definition}[Strong/weak separation of distributions]
Suppose $\PP_n$ and $\QQ_n$ are distributions on $\RR^N$ for some $N = N_n$. A polynomial $f_n: \RR^N \to \RR$ is said to \emph{strongly separate} $\PP$ and $\QQ$ if, as $n \to \infty$,
\[ \sqrt{\max\left\{\Var_\QQ[f], \Var_\PP[f]\right\}} = o\left(\left|\EE_\PP[f] - \EE_\QQ[f]\right|\right), \]
and \emph{weakly separate} $\PP$ and $\QQ$ if
\[ \sqrt{\max\left\{\Var_\QQ[f], \Var_\PP[f]\right\}} = O\left(\left|\EE_\PP[f] - \EE_\QQ[f]\right|\right). \]
\end{definition}

Note that strong separation implies that $\PP$ and $\QQ$ can be strongly distinguished by thresholding the value of the polynomial $f$. Weak separation implies that the output of $f$ can be used to distinguish better than random guessing; see~\cite[Proposition~6.1]{fp}.

In our case, the input variables will take values in $\{\pm 1\}$ and so the polynomial $f$ can be multilinear without loss of generality.

If all degree-$D$ polynomials \emph{fail} to strongly separate $\PP$ and $\QQ$ for some $D = \omega(\log n)$, we say the testing problem is ``low-degree hard.'' As explained in the introduction, this can be viewed as evidence for inherent computational hardness of strong distinguishing.

Proofs that rule out strong or weak separation typically proceed by bounding the {\em advantage}, defined below:
\begin{equation}\label{eq:adv}
\Adv_{\le D}(\PP,\QQ) \coloneqq \sup_{f \in \RR[Y]_{\le D}} \frac{\E_\PP[f]}{\sqrt{\E_\QQ[f^2]}},
\end{equation}
where $\RR[Y]_{\le D}$ denotes the set of polynomials $\RR^N \to \RR$ of degree (at most) $D$. It is well known that $\Adv_{\le D}$ also admits a characterization as the \emph{norm of the low-degree likelihood ratio}; see~\cite{hopkins-thesis,ld-notes}. If $\Adv_{\le D} = O(1)$ then strong separation is impossible, and if $\Adv_{\le D} = 1 + o(1)$ then weak separation is impossible (see Lemma~\ref{lem:adv-sep}).

\subsection{Recovery}

\subsubsection{Models}

The primary objective of this section will be to understand the recovery problem in two related models for planted coloring and planted partial-coloring. As explained in the introduction, the complement of a $q$-colorable graph is partitioned into $q$ cliques. To fix notation and compare with the standard planted clique model, we will take the clique perspective here. Thus we study the problem of multiple cliques planted in $G(n,1/2)$.

The first model $\MC(n,q)$ (``multiple cliques'') corresponds to a true planted coloring, i.e., the cliques partition the entire graph.

\begin{definition}\label{def:MC}
In the model $\MC(n,q)$, we observe an $n$-vertex graph where each vertex is independently assigned a uniformly random label from $[q] \coloneqq \{1,2,\ldots,q\}$. Vertices with the same label are always connected, and vertices with different label are connected with probability $1/2$. Given the graph, the goal is to exactly recover the clique partition with probability $1-o(1)$ as $n \to \infty$, where $q = q_n$ may scale with $n$.
\end{definition}

The next model is a variation for partial coloring, i.e., the cliques do not partition the entire graph. For technical convenience, the cliques in this model have exactly the same size, unlike $\MC(n,q)$.

\begin{definition}
In the model $\MC(n,q,\delta)$, we observe an $n$-vertex graph where $(1-\delta)n$ vertices are partitioned into $q$ cliques, each of size exactly $k \coloneqq (1-\delta)n/q$ (which we assume is an integer). Two vertices in the same clique are always connected, and all remaining edges occur independently with probability $1/2$. Given the graph, the goal is to exactly recover the clique partition (and identify the non-clique vertices) with probability $1-o(1)$ as $n \to \infty$, where the parameters $q = q_n$ and $\delta = \delta_n$ may scale with $n$.
\end{definition}

\subsubsection{Hardness of planted partial-coloring via reduction}
\label{sec:reduction}

We now consider the recovery problem in $\MC(n,q,\delta)$. First, we observe that a simple algorithm based on examining degrees and common neighbors can exactly recover the cliques when $k \gg \sqrt{n}$. This matches (up to log factors) the best known algorithms for recovering a single planted $k$-clique in $G(n,1/2)$.

\begin{restatable}[Upper bound]{theorem}{thmrecoveryupper}
\label{thm:exact_MC_algo}
If $q,\delta$ scale with $n$ such that $k \coloneqq (1-\delta)n/q = \omega(\sqrt{n \log n})$ then there is a polynomial-time algorithm achieving exact recovery w.h.p.\ in $\MC(n,q,\delta)$.
\end{restatable}

We next show a matching lower bound: computational hardness of recovering the cliques when $k \ll \sqrt{n}$. This result will be conditional on the \emph{Planted Clique Hypothesis}, a conjecture that is commonly used as the basis for deducing average-case hardness results. In the \emph{planted clique model} $\PC(N,K)$, an $N$-vertex graph has a clique on $K$ vertices, and all other edges occur independently with probability $1/2$. The following version of the conjecture appears, for instance, as Conjecture~2.1 in~\cite{BBH-reductions}.

\begin{conjecture}[Planted Clique Hypothesis]
\label{conj:pc}
If $K = K_N$ scales as $K \le N^{1/2 - \Omega(1)}$ then no sequence of randomized polynomial-time algorithms $B_N$ can strongly distinguish (Definition~\ref{def:strong-det}) between $\PC(N,K)$ and $G(N,1/2)$.
\end{conjecture}

Assuming this conjecture, we have the following hardness result for $\MC(n,q,\delta)$.

\begin{restatable}[Lower bound]{theorem}{thmrecoverylower}
\label{thm:reduction}
Assume the Planted Clique Hypothesis (Conjecture~\ref{conj:pc}). If $q,\delta$ scale with $n$ such that $k \coloneqq (1-\delta)n/q$ satisfies $(2+\Omega(1)) \log_2 n \le k \le (\delta n)^{1/2 - \Omega(1)}$ then no sequence of randomized polynomial-time algorithms $A_n$ achieves exact recovery w.h.p.\ in $\MC(n,q,\delta)$.
\end{restatable}

\noindent The condition $(2+\Omega(1)) \log_2 n \le k$ is natural because $2\log_2 n$ is the size of the maximum clique in $G(n,1/2)$. To satisfy the condition $k \le (\delta n)^{1/2-\Omega(1)}$, it suffices to have $k = n^{\frac{1}{2} - \Omega(1)}$ and $\delta = n^{-o(1)}$.

The reduction which proves Theorem~\ref{thm:reduction} is very simple but (to our knowledge) has not appeared before in the literature. Intuitively, the idea is the following: in the multiple cliques model, even if an oracle were to reveal the positions of all cliques but one, the remaining problem is still a hard instance of planted clique.

\begin{remark}
We note that the Planted Clique Hypothesis also implies hardness of detecting a constant number of planted $k$-cliques in $G(n,1/2)$ when $k \ll \sqrt{n}$. The idea is to first show by reduction from planted clique that distinguishing between $q$ planted cliques and $(q+1)$ planted cliques is hard; the reduction is simply to add $q$ new cliques (on new vertices). Then the classical ``hybrid argument'' implies that distinguishing between $0$ and $q$ cliques is hard for any constant $q$. (We thank Guy Bresler for pointing out this argument.)
\end{remark}

\paragraph{Detection-recovery gap.}

In the standard planted clique model (with a single clique), $k \sim \sqrt{n}$ is the best known threshold for both efficiently recovering the clique and efficiently ``detecting'' it, i.e., distinguishing the planted clique model from $G(n,1/2)$. While we have shown that adding more cliques does not make recovery any easier, it certainly does make detection easier. For instance, in the extreme case where the cliques cover the whole graph, the total edge count strongly distinguishes $\MC(n,q)$ from $G(n,1/2)$ provided $q = o(n)$. Thus, the multiple cliques problem exhibits a ``detection-recovery gap'' that is not present in the single clique case.
 
We remark that our reduction is a rare (perhaps unique?)\ example where a detection-recovery gap has been established based on the Planted Clique Hypothesis. For instance, the prior work~\cite{BBH-reductions} on various planted matrix and graph problems was only able to establish hardness of recovery in a regime where detection is easy if reducing from some starting problem (not planted clique) that is already conjectured to have a detection-recovery gap. While~\cite{CLR-reduction} claims to overcome this by reducing from planted clique to planted submatrix recovery, the argument is incorrect.\footnote{On pg 21-22 of~\cite{CLR-reduction} (arXiv~v2), the bootstrapping construction in Eq.~(42) does not actually produce an instance of the submatrix model because the entries of the noise matrix are not mutually independent. An issue occurs near the top of pg 22, where pairwise independence does not imply mutual independence. The reduction does show hardness of some non-standard submatrix model where the noise entries are not mutually independent.}
 
Finally we note that the notion of a ``detection-recovery gap'' is arguably somewhat artificial in that it assumes we have chosen one ``canonical'' testing problem to associate with the recovery problem (a perspective we are avoiding in this paper).

\subsubsection{Testing $q$-colorability versus $(q+\ell)$-colorability}
\label{sec:q-vs-ql}

The results of the previous section do not quite cover the case of a true coloring, i.e., where the cliques partition the entire graph. In this case, exact recovery remains easy when $k \coloneqq n/q \gg \sqrt{n}$, and we expect it to be hard when $k \ll \sqrt{n}$; however, we do not know how to establish this via reduction from planted clique. We will instead follow strategy (I) from the introduction: we fix $\PP = \MC(n,q)$ and our goal is to design a null distribution $\QQ$ such that w.h.p.\ $G \sim \QQ$ is not $q$-colorable (or rather, its complement is not), and distinguishing $\PP$ versus $\QQ$ is low-degree hard. Once we have achieved this goal, this gives an indirect two-stage argument for hardness of recovery: the low-degree hardness leads us to conjecture that no poly-time algorithm can distinguish $\PP$ from $\QQ$, and this conjecture (if true) formally implies that no poly-time algorithm can recover the cliques in $\PP$.

Perhaps the first natural attempt is to choose $\QQ = G(n,1/2)$. However, this will not suffice, as $G(n,1/2)$ is too easy to distinguish from $\MC(n,q)$ due to the detection-recovery gap discussed in the previous section. Instead, we will choose $\QQ = \MC(n,q+1)$, which w.h.p.\ is not $q$-colorable for $q \leq \Omega(n/\log n)$; see Appendix~\ref{app:coloring}.
We will show that testing $\PP = \MC(n,q)$ versus $\QQ = \MC(n,q+1)$ is low-degree hard when $k \coloneqq n/q \ll \sqrt{n}$. As discussed above, this suggests hardness of exact recovery in $\MC(n,q)$ when $k \ll \sqrt{n}$.

We will in fact consider a slightly more general testing problem: $\PP = \MC(n,q)$ versus $\QQ = \MC(n,q+\ell)$ for some $\ell \ge 1$ (which may scale with $n$). This generality will not cost us much, and we feel it is a question of possible independent interest. The following results establish that (in the low-degree framework) this problem is easy when $q^2 \ll \ell n$ and hard when $q^2 \gg \ell n$.

\begin{restatable}[Upper bound]{theorem}{thmtestingupper}
\label{thm:testing-upper}
If $q,\ell$ scale with $n$ such that $1 \le q < q+\ell \le n$ and $q^2 = o(\ell n)$ then there is a degree-1 polynomial achieving strong separation between $\PP = \MC(n,q)$ and $\QQ = \MC(n,q+\ell)$.
\end{restatable}

\vspace{-3pt}

\begin{restatable}[Lower bound]{theorem}{thmtestinglower}\label{thm:testing-lower}
Fix an arbitrary constant $\epsilon > 0$, not depending on $n$. If $q,\ell$ scale with $n$ such that $1 \le q < q+\ell \le n$ and $q^2 \ge \ell n^{1+\epsilon}$ then there is no degree-$o(\log n/\log\log n)^2$ polynomial achieving weak separation between $\PP = \MC(n,q)$ and $\QQ = \MC(n,q+\ell)$.
\end{restatable}

\paragraph{Testing planted versus planted.}

On a technical level, this result differs from nearly all existing low-degree lower bounds because here we are testing between two different ``planted'' distributions. In contrast, most prior work has considered testing between some planted distribution and an i.i.d.\ null distribution, which is much easier to analyze. The first ``planted-versus-planted'' low-degree lower bounds were given recently by~\cite{planted-planted}, based on a technique developed by~\cite{SW-recovery}. Our proof is based on similar ideas, but differs from~\cite{planted-planted} on a technical level; the bounds for dense subgraph problems in~\cite{planted-planted} do not work when the subgraph is extremely dense (e.g., a clique), and so we use a somewhat different variation of the argument.

The key technical challenge is that, since $\QQ$ is not i.i.d., we do not know an orthogonal basis of polynomials (w.r.t.\ $\QQ$) that is convenient to work with. Proposition~\ref{prop:adv-bound} overcomes this, showing that it suffices to control certain recursively-defined quantities $w_\alpha$. This generalizes the standard approach; see Remark~\ref{rem:gen}. Similarly to~\cite{planted-planted}, the quantities $w_\alpha$ turn out to have a convenient multiplicative property (Lemma~\ref{lem:multiplicative}) which helps in the analysis.

We note that an alternative form of evidence for hardness of our original recovery problem would be to directly formulate a low-degree recovery question in the style of~\cite{SW-recovery}, but we have chosen to instead investigate the quiet planting approach.

\subsection{Refutation}

A common framework for studying the average-case complexity of refutation problems is to prove lower bounds against the sum-of-squares (SoS) hierarchy, a powerful class of methods based on semi-definite programming. For the problem of refuting $q$-colorability, a particular SoS formulation is known to fail when $q \gg \sqrt{n}$~\cite{KM21}; however, it remains open to characterize the more canonical (and potentially stronger) SoS SDP which has equality constraints instead of inequalities (see Section~1.5 of~\cite{KM21}).

In this paper, we formulate an alternative type of refutation lower bound based directly on low-degree polynomials, which complements the SoS approach. Some advantages of the new formulation are its simplicity, and the fact that (unlike SoS) there is no ambiguity in the choice of SDP relaxation; we only need to specify how our input is encoded as real-valued variables. To our knowledge, there are no formal implications in either direction between SoS lower bounds and our new framework. Like SoS, our framework captures spectral methods (as illustrated by the proof of Theorem~\ref{thm:ref-upper} below), a powerful class of refutation algorithms which give the best known poly-time algorithms for a wide variety of average-case refutation tasks.

We note that some prior work has used low-degree lower bounds to give evidence for hardness of refutation, via a two-stage argument that first gives a polynomial-time reduction from a testing problem to refutation~\cite{sk-cert,nonneg-pca}. Our new framework is similar in spirit but more direct, as we define for the first time a notion of what it means for a polynomial to solve a refutation problem (Definition~\ref{def:ref-sep}).

\subsubsection{Framework for low-degree refutation}

We will now define a notion (Definition~\ref{def:ref-sep}) of what it means for a polynomial to refute a property $\mathcal{R} \subseteq \RR^N$ (e.g., the set of $q$-colorable graphs $X \in \{\pm 1\}^{\binom{n}{2}}$) over a distribution $\QQ$ (e.g., $G(n,1/2)$). We will later argue that this definition is reasonable in that it indeed implies a solution to the refutation problem (Proposition~\ref{prop:sep-implies-ref}). We also illustrate that our definition captures spectral methods, a powerful class of refutation algorithms (see the proof of Theorem~\ref{thm:ref-upper}).

\begin{definition}[Strong/weak separation of a distribution and property]
\label{def:ref-sep}
Suppose $\QQ_n$ is a distribution on $\RR^N$ for some $N = N_n$, and suppose $\mathcal{R} = \mathcal{R}_n \subseteq \RR^N$. A polynomial $f_n: \RR^N \to \RR$ is said to \emph{strongly separate} $\QQ$ and $\mathcal{R}$ if
\[ f(X) \ge 1 \;\;\; \forall X \in \mathcal{R} \qquad\text{and}\qquad \EE_\QQ[f^2] = o(1), \]
and \emph{weakly separate} $\QQ$ and $\mathcal{R}$ if
\[ f(X) \ge 1 \;\;\; \forall X \in \mathcal{R} \qquad\text{and}\qquad \EE_\QQ[f] = 0, \;\; \EE_\QQ[f^2] = O(1). \]
\end{definition}

\begin{remark}
The requirement $\EE_\QQ[f] = 0$ can optionally be added to the definition of strong separation: if $f = f_n$ satisfies the original definition it can be shifted and scaled to satisfy the modified one.

More generally, one could define separation to mean there exists $B = B_n > \EE_\QQ[f]$ such that $f(X) \ge B$ for all $X \in \mathcal{R}$, and $\sqrt{\Var_\QQ[f]}$ is either $o(B - \EE_\QQ[f])$ (for strong separation) or $O(B - \EE_\QQ[f])$ (for weak separation). This is equivalent in the sense that if $f = f_n$ satisfies the original definition it also satisfies the new one with $B=1$, and if $f$ satisfies the new definition it can be shifted and scaled to satisfy the original one.
\end{remark}

As we see next, strong and weak separation are natural sufficient conditions for refuting $\mathcal{R}$ with high probability or constant probability (respectively) by evaluating $f$.

\begin{proposition}
\label{prop:sep-implies-ref}
Suppose $f$ strongly (or weakly, respectively) separates $\QQ$ and $\mathcal{R}$. Define a refutation algorithm that, on input $X \in \RR^N$, outputs NO if $f(X) < 1$ and outputs MAYBE otherwise. Then this algorithm has the guarantee that (1) whenever it outputs NO, $X \notin \mathcal{R}$, and (2) when $X \sim \QQ$, the output is NO with probability $1-o(1)$ (or $\Omega(1)$, respectively).
\end{proposition}

\begin{proof}
Guarantee (1) is immediate from the property $f(X) \ge 1$ for all $X \in \mathcal{R}$. For strong separation, (2) follows because by Markov's inequality, $\EE[f^2] = o(1)$ implies that $|f(X)| < 1$ with probability $1-o(1)$. It remains to verify (2) for weak separation: letting $\EE_\QQ[f^2] \le C$ and $p \coloneqq \Pr_\QQ\{f(X) < 1\}$,
\[ 0 = \EE[f] \ge 1 \cdot \Pr\{f \ge 1\} + \EE[f \cdot \One_{f < 1}] \ge (1-p) - \sqrt{\EE[f^2]} \cdot \sqrt{p} \ge 1-p-C\sqrt{p} \ge 1 - (C+1)\sqrt{p}, \]
implying $p \ge 1/(C+1)^2$.
\end{proof}

In line with strategy (II) from the introduction, one way to rule out strong (or weak) separation is to construct a planted distribution and bound the quantity $\Adv_{\le D}$ defined in~\eqref{eq:adv}.

\begin{proposition}\label{prop:strat-II}
Suppose that on an infinite subsequence of $n$ values we have a distribution $\PP = \PP_n$ supported on $\mathcal{R}$. If $\Adv_{\le D}(\PP,\QQ) = O(1)$ (respectively, $1+o(1)$) for some $D = D_n$, then no degree-$D$ polynomial strongly (resp., weakly) separates $\QQ$ and $\mathcal{R}$.
\end{proposition}

\begin{proof}
Since $\PP$ is supported on $\mathcal{R}$, the separation condition implies $\EE_\PP[f] \ge 1$. The proof is now nearly identical to that of Lemma~\ref{lem:adv-sep}.
\end{proof}

\begin{remark}
We note that for the well-studied problem of refuting a single $k$-clique in $G(n,1/2)$, existing work implies sharp upper and lower bounds in our new framework. For the lower bound, let $\PP$ be the standard planted $k$-clique model and combine Proposition~\ref{prop:strat-II} with the low-degree analysis of planted clique~\cite[Section~2.4]{hopkins-thesis} to conclude: if $k \le n^{1/2-\epsilon}$ for a constant $\epsilon > 0$ then no degree-$o(\log n/\log \log n)^2$ polynomial weakly separates $G(n,1/2)$ from the property of containing a $k$-clique. The upper bound follows from the proof of Theorem~\ref{thm:ref-upper} below: if $k \ge 2.1 \sqrt{n}$ then there is an $O(\log n)$-degree polynomial that strongly separates $G(n,1/2)$ from the property of containing a $k$-clique.
\end{remark}

\subsubsection{Low-degree refutation of $q$-colorability}

We now apply the framework from the previous section to the problem of refuting $q$-colorability in $G(n,1/2)$. Throughout, we represent graphs as elements of $\{\pm 1\}^{\binom{n}{2}}$ as usual, take $\QQ = G(n,1/2)$, and use $\mathcal{R}_q \subseteq \{\pm 1\}^{\binom{n}{2}}$ to denote the property of $q$-colorability (i.e., the set of graphs that are $q$-colorable).

First, we give an upper bound: low-degree polynomials can refute $q$-colorability for $q \lesssim \sqrt{n}$. The proof proceeds by taking a standard spectral refutation algorithm (based on the maximum eigenvalue of the adjacency matrix) and approximating it by a polynomial.

\begin{restatable}[Upper bound]{theorem}{thmrefupper}
\label{thm:ref-upper}
Suppose $q \le b \sqrt{n}$ for a constant $b < 1/2$ (not depending on $n$). Then there exists a constant $C = C(b) > 0$ and a polynomial $f = f_n$ of degree at most $C \log n$ that strongly separates $G(n,1/2)$ and $\mathcal{R}_q$.
\end{restatable}

\noindent We also give a lower bound: no low-degree polynomial can refute $q$-colorability for $q \gg n^{2/3}$. Note there is a gap between our upper and lower bounds, and we leave closing this gap as an interesting direction for future work.

\begin{restatable}[Lower bound]{theorem}{thmreflower}\label{thm:ref-lower}
If $q \ge n^{2/3 + \epsilon}$ for a constant $\epsilon > 0$, then no degree-$o(\log n/\log\log n)^2$ polynomial weakly separates $G(n,1/2)$ and $\mathcal{R}_q$.
\end{restatable}

\noindent The proof of the lower bound will use Proposition~\ref{prop:strat-II}, which is a rigorous incarnation of strategy (II) from the introduction. In other words, our goal is to construct a planted distribution $\PP$ supported on $q$-colorable graphs that is hard to distinguish from $\QQ = G(n,1/2)$ in the sense $\Adv_{\le D}(\PP,\QQ) = 1+o(1)$.

Constructing this planted distribution is non-trivial. The naive choice would be the ``canonical'' planted model $\MC(n,q)$ (or rather, its complement), but this is not a good choice because it can be easily distinguished from $G(n,1/2)$ by counting the total number of edges whenever $q \ll n$. A next attempt is to modify $\MC(n,q)$ to have a slightly lower probability for non-clique edges so as to correct the total edge count. This gives a quieter planting that is hard to distinguish from $\QQ$ when $q \gg n^{3/4}$, but easy when $q \ll n^{3/4}$ by counting signed triangles (each of the $\binom{n}{3}$ triangles in the complete graph counts for $+1$ if an even number of its edges are present or $-1$ if an odd number are present). Our final construction, defined below, that reaches the threshold $q \sim n^{2/3}$, is more complicated and involves planting both cliques and independent sets.

\begin{definition}[Quiet planting for $q \gg n^{2/3}$]
\label{def:planting}
Suppose $n,q$ are positive integers. To each of the $n$ vertices, independently assign a label $(a,b) \in [q] \times [q]$ uniformly at random. Conditioned on the labels, do the following independently for each pair of distinct vertices $\{u,v\}$: denote the two vertex labels by $(a_1,b_1)$ and $(a_2,b_2)$; if $a_1 = a_2$ then do not include the edge $(u,v)$; if $a_1 \ne a_2$ and $b_1 = b_2$ then include the edge $(u,v)$; otherwise include the edge $(u,v)$ with probability $1/2$.
\end{definition}

\noindent Note that all the vertices with a given $a$ value form an independent set, and thus the distribution is supported on $q$-colorable graphs. Also, the vertices with a given $b$ value nearly form a clique, aside from the non-edges required for the independent sets. In the proof of Theorem~\ref{thm:ref-lower}, we show that this distribution is low-degree indistinguishable from $G(n,1/2)$ when $q \gg n^{2/3}$. Our analysis of this distribution is tight, as the count of signed 4-cycles distinguishes it from $G(n,1/2)$ when $q \ll n^{2/3}$.

Although we have not proven it, we expect the true threshold for low-degree refutation of colorability to be $q \sim \sqrt{n}$.

\begin{conjecture}\label{conj:quiet}
Fix an arbitrary $\epsilon > 0$, not depending on $n$. If $q \ge n^{1/2 + \epsilon}$ then no degree-$D$ polynomial weakly separates $\QQ = G(n,1/2)$ and $\mathcal{R}_q$, for some $D = \omega(\log n)$.
\end{conjecture}

\subsubsection{Completeness of the quiet planting approach}

A natural approach to prove Conjecture~\ref{conj:quiet} would be to construct a quieter planted distribution $\PP$ that is supported on $q$-colorable graphs but hard to distinguish from $G(n,1/2)$ when $q \gg \sqrt{n}$. One might worry, however, that this may not even be possible: conceivably, such a planted distribution might not exist, even if the true low-degree refutation threshold is at $q \sim \sqrt{n}$ like we expect. If this were the case, we would need to find an alternative approach to prove the conjecture without relying on quiet planting.

We show in high generality that the hypothetical scenario above actually cannot occur: for every low-degree hard refutation problem, there is a planted distribution that can be used to prove its hardness. Put another way, Conjecture~\ref{conj:quiet} is \emph{equivalent} to the existence of a quiet planted distribution for $q \gg \sqrt{n}$.

\begin{restatable}{theorem}{thmrefduality}
Fix sequences $N = N_n$, $D = D_n$, $\QQ = \QQ_n$ a distribution on $\RR^N$, and $\mathcal{R} = \mathcal{R}_n \subseteq \RR^N$. Assume that for each $n$, $\QQ$ is supported on a finite set and $\mathcal{R}$ is a finite set (but the cardinality of these sets may depend on $n$).
The following are equivalent:
\begin{itemize}
\item[(1)] No degree-$D$ polynomial strongly separates $\QQ$ and $\mathcal{R}$.
\item[(2)] For an infinite subsequence of $n$ values, there exists a distribution $\PP = \PP_n$ supported on $\mathcal{R}$ such that $\Adv_{\le D}(\PP,\QQ) = O(1)$.
\end{itemize}
Similarly, the following are equivalent:
\begin{itemize}
\item[(1)] No degree-$D$ polynomial weakly separates $\QQ$ and $\mathcal{R}$.
\item[(2)] For an infinite subsequence of $n$ values, there exists a distribution $\PP = \PP_n$ supported on $\mathcal{R}$ such that $\Adv_{\le D}(\PP,\QQ) = 1+o(1)$.
\end{itemize}
\end{restatable}

\noindent Note that we have already shown that (2) implies (1); see Proposition~\ref{prop:strat-II}. The proof that (1) implies (2) uses von Neumann's min-max principle.

\begin{remark}
We have assumed $\mathrm{supp}(\QQ)$ and $\mathcal{R}$ are finite (the relevant setting for $q$-coloring) to simplify the analytic conditions needed for the min-max principle, but these assumptions can be relaxed; see Remark~\ref{rem:compact}.
\end{remark}

Adapted to the context of $q$-coloring, while formally we do not know whether there exists a low-degree polynomial to refute $q$-coloring when $q \gg \sqrt{n}$, it would be surprising in light of the sum-of-squares lower bound~\cite{KM21} for refuting $\tilde{O}(\sqrt{n})$-colorability of $G(n,1/2)$. Hence, we interpret this argument as suggesting the existence of a computationally quiet planted $q$-coloring for $G(n,1/2)$ when $q \approx \sqrt{n}$ even though we do not know an explicit construction of such a distribution. If this construction were known, it may allow for SoS lower bounds in stronger SDP formulations to be proved via the \emph{pseudo-calibration}~\cite{BHKK+16} approach.

\section{Recovering Multiple Cliques}

\subsection{Upper Bound}

We restate the theorem for the reader's convenience.

\thmrecoveryupper*

\begin{proof}
We will use the following standard version of Bernstein's inequality: for independent random variables $X_1,\ldots,X_n$ satisfying $\EE[X_i] = 0$ and $|X_i| \le M$ almost surely, we have for any $t \ge 0$ that
\[ \Pr\left(\sum_{i=1}^n X_i \ge t\right) \le \exp\left(-\frac{\frac{1}{2}t^2}{\sum_{i=1}^n \Var(X_i) + \frac{1}{3}Mt}\right). \]

Fix an arbitrary sequence $\alpha_n = \omega(1)$. The degree $d_i$ of a non-clique vertex $i$ has a binomial distribution $d_i \sim \Bin(n-1,1/2)$, which by Bernstein's inequality satisfies $d_i \le \frac{n}{2} + \alpha\sqrt{n \log n}$ with probability $1-n^{-\omega(1)}$. On the other hand, a clique vertex $i$ has degree $d_i \sim (k-1) + \Bin(n-k,1/2)$, which by Bernstein's inequality satisfies $d_i \ge \frac{n+k}{2} - \alpha\sqrt{n \log n}$ with probability $1-n^{-\omega(1)}$. By thresholding degrees, this lets us perfectly classify the non-clique vertices with probability $1-o(1)$, provided $k = \omega(\sqrt{n \log n})$.

It remains to partition the clique vertices. If vertices $i,j$ are in different cliques, their number of common neighbors is $d_{ij} \sim \Bin(2(k-1),1/2) + \Bin(n-2k,1/4)$, which satisfies $d_{ij} \le \frac{n}{4} + \frac{k}{2} + \alpha\sqrt{n \log n}$ with probability $1-n^{-\omega(1)}$. If vertices $i,j$ instead belong to the same clique, their number of common neighbors is $d_{ij} \sim (k-2) + \Bin(n-k,1/4)$, which satisfies $d_{ij} \ge \frac{n}{4} + \frac{3k}{4} - \alpha\sqrt{n \log n}$ with probability $1-n^{-\omega(1)}$. By thresholding common neighbors, this allows us to exactly recover the clique partition with probability $1-o(1)$, again provided $k = \omega(\sqrt{n \log n})$.
\end{proof}

\subsection{Lower Bound via Reduction}

We restate the theorem for the reader's convenience.

\thmrecoverylower*

\begin{proof}
Let $q,\delta$ scale as prescribed. Assume for the sake of contradiction that an algorithm $A_n$ achieves exact recovery in $\MC(n,q,\delta)$. Let $K = k = (1-\delta)n/q$ and $N = K + \delta n$. Note that as $n \to \infty$ we have $N \to \infty$ because
\[ N \ge K = k \ge (2+\Omega(1)) \log_2 n \to \infty, \]
and also $K \le N^{1/2 - \Omega(1)}$ because
\[ K = k \le (\delta n)^{\frac{1}{2}-\Omega(1)} \le N^{\frac{1}{2}-\Omega(1)}. \]
We will give an algorithm $B_N$ achieving strong detection between $G(N,1/2)$ and $\PC(N,K)$, contradicting the planted clique conjecture.

The algorithm $B_N$ works as follows. Given an $N$-vertex graph, add $(q-1)k$ additional vertices (bringing the total to $n$), partitioned into $q-1$ cliques each of size $k$. Add all other edges (both among the new vertices and between the old and new vertices) independently with probability $1/2$. Now run $A_n$ on the resulting graph. If it finds $q$ disjoint cliques of size $k$ and one of these cliques lies within the original $N$ vertices, output ``$\PC(N,K)$''; otherwise, output ``$G(N,1/2)$.''

To argue correctness of $B_N$, first suppose the input came from $\PC(N,K)$. Then the $n$-vertex graph produced is exactly a sample from $\MC(n,q,\delta)$, and so $A_n$ must correctly identify all the cliques with probability $1-o(1)$, leading $B_N$ to correctly answer ``$\PC(N,k)$.'' Now suppose instead that the input to $B_N$ came from $G(N,1/2)$. Due to the assumption $k \ge (2+\Omega(1)) \log_2 n \ge (2+\Omega(1)) \log_2 N$, with probability $1-o(1)$ there is no $k$-clique within the original $N$ vertices, in which case $B_N$ must correctly answer ``$G(N,1/2)$.''
\end{proof}

\section{Testing the Number of Cliques}

\subsection{Upper Bound}

We restate the theorem for the reader's convenience.

\thmtestingupper*

\begin{proof}
Let $f$ be the degree-1 polynomial that counts the total number of signed edges in the graph: $f(Y) = \sum_{1 \le i < j \le n} Y_{ij}$, where recall $Y_{ij} \in \{\pm 1\}$. Using linearity of expectation,
\[ \EE_{Y \sim \MC(n,q)} f(Y) = \binom{n}{2} \frac{1}{q} \]
and so
\begin{equation}\label{eq:first-moment-diff}
\left|\EE_{Y \sim \PP} f(Y) - \EE_{Y \sim \QQ} f(Y)\right| = \binom{n}{2}\left(\frac{1}{q} - \frac{1}{q+\ell}\right) = \binom{n}{2} \frac{\ell}{q(q+\ell)}.
\end{equation}
For the second moment,
\[ \EE_{Y \sim \MC(n,q)} f(Y)^2 = \sum_{i < j} \sum_{i' < j'} \EE[Y_{ij} Y_{i'j'}]. \]
There are a few different terms to consider depending on how the edges $(i,j)$ and $(i',j')$ interact.
\begin{itemize}
\item If $(i,j) = (i',j')$ then $\EE[Y_{ij}Y_{i'j'}] = \EE[Y_{ij}^2] = 1$.
\item If $(i,j)$ and $(i',j')$ have no vertices in common then $Y_{ij}$ and $Y_{i'j'}$ are independent, and so $\EE[Y_{ij}Y_{i'j'}] = \EE[Y_{ij}]\EE[Y_{i'j'}] = \frac{1}{q^2}$.
\item If $(i,j)$ and $(i',j')$ have one vertex in common then we again have that $Y_{ij}$ and $Y_{i'j'}$ are independent: if say $i = i'$ then the event that $i,j$ have the same label is independent from the event that $i,j'$ have the same label, due to symmetry among the possible labels for $i$. Therefore $\EE[Y_{ij}Y_{i'j'}] = \frac{1}{q^2}$.
\end{itemize}
Putting it together,
\[ \EE_{Y \sim \MC(n,q)} f(Y)^2 = \binom{n}{2} \cdot 1 + \binom{n}{2}\left[\binom{n}{2}-1\right] \cdot \frac{1}{q^2} \]
and so
\begin{align}
\Var_{Y \sim \MC(n,q)} f(Y) &= \binom{n}{2} \cdot 1 + \binom{n}{2}\left[\binom{n}{2}-1\right] \cdot \frac{1}{q^2} - \left[\binom{n}{2} \frac{1}{q}\right]^2 \nonumber \\
&= \binom{n}{2}\left(1 - \frac{1}{q^2}\right) \nonumber \\
&\le \binom{n}{2}.
\label{eq:variance}
\end{align}
Combining~\eqref{eq:first-moment-diff} and~\eqref{eq:variance}, $f$ achieves strong separation provided
\[ \sqrt{\binom{n}{2}} = o\left(\binom{n}{2} \frac{\ell}{q(q+\ell)}\right), \qquad \text{i.e.,} \qquad q\left(\frac{q}{\ell}+1\right) = o(n). \]
It therefore suffices to have $q = o(n)$ and $q^2 = o(\ell n)$. Note that $q=o(n)$ is implied by $q^2 = o(\ell n)$ together with $\ell \le n$.
\end{proof}

\subsection{Lower Bound}
\label{sec:pf-testing-lower}

We restate the theorem for the reader's convenience.

\thmtestinglower*

\subsubsection{Proof overview}

We first perform a standard manipulation, showing that it suffices to bound the quantity $\Adv_{\le D}$.

\begin{lemma}\label{lem:adv-sep}
Let $\PP = \PP_n$ and $\QQ = \QQ_n$ be distributions on $\RR^N$ for some $N = N_n$. For some $D = D_n$, let $\RR[Y]_{\le D}$ denote the set of polynomials $\RR^N \to \RR$ of degree (at most) $D$. If
\[ \Adv_{\le D}(\PP,\QQ) \coloneqq \sup_{f \in \RR[Y]_{\le D}} \frac{\E_\PP[f]}{\sqrt{\E_\QQ[f^2]}} = 1+o(1), \]
then no degree-$D$ polynomial $f: \RR^N \to \RR$ weakly separates $\PP$ and $\QQ$. Similarly, if $\Adv_{\le D}(\PP,\QQ) = O(1)$ then no degree-$D$ polynomial strongly separates $\PP$ and $\QQ$.
\end{lemma}

\noindent It is always the case that $\Adv_{\le D} \ge 1$, by taking $f = 1$.

\begin{proof}
Assume for the sake of contradiction that some degree-$D$ polynomial $g: \RR^N \to \RR$ weakly separates $\PP$ and $\QQ$. By shifting and scaling, we can assume without loss of generality that $\EE_\QQ[g] = 0$ and $\EE_\PP[g] = 1$. For sufficiently large $n$, weak separation guarantees $\Var_\QQ[g] = \EE_\QQ[g^2] \le C$ for some constant $C > 0$. Define $f = g + C$ and compute
\[ \frac{\EE_\PP[f]}{\sqrt{\EE_\QQ[f^2]}} = \frac{1+C}{\sqrt{\EE_\QQ[g^2]+C^2}} \ge \frac{1+C}{\sqrt{C+C^2}} = \sqrt{\frac{1+C}{C}}, \]
which is a constant strictly greater than $1$, contradicting $\Adv_{\le D} = 1+o(1)$. The proof for strong separation is similar, now with $C = o(1)$.
\end{proof}

A key ingredient in the proof will be an upper bound on $\Adv_{\le D}$ in the following generic setting (of which our problem is a special case). Suppose $\QQ$ takes the form $Y = X \vee Z$ where $X, Z \in \{\pm 1\}^N$ with ``noise'' $Z$ i.i.d.\ Rademacher and ``signal'' $X$ having an arbitrary distribution (independent from $Z$), and $\vee$ denotes entrywise maximum. (In our case $N = \binom{n}{2}$ and $X$ is the $\pm 1$-valued indicator for clique edges.)

\begin{proposition}\label{prop:adv-bound}
Suppose $\QQ$ takes the form $Y = X \vee Z$ as described above and $\PP$ is any distribution on $\{\pm 1\}^N$. For $\alpha,\beta \subseteq [N]$, define
\[ c_\alpha = \EE_{Y \sim \PP}[Y^\alpha] \coloneqq \EE_{Y \sim \PP} \prod_{i \in \alpha} Y_i \]
and
\[ M_{\beta\alpha} = \Pr_X(\alpha \setminus X = \beta). \]
Here and throughout, we abuse notation and use $X$ to refer to the set $\{i \in [N] \,:\, X_i = 1\}$.

Suppose $M_{\alpha\alpha} > 0$ for all $|\alpha| \le D$. Then
\begin{equation}\label{eq:adv-bound}
\Adv_{\le D}^2 \le \sum_{\alpha \subseteq [N],\, |\alpha| \le D} w_\alpha^2
\end{equation}
where $w_\alpha$ is defined recursively by
\[ w_\alpha = \frac{1}{M_{\alpha\alpha}}\left(c_\alpha - \sum_{\beta \subsetneq \alpha} w_\beta M_{\beta\alpha}\right). \]
No explicit base case is needed for the recursion above, but one can think of $w_\emptyset = 1$ as the base case.
\end{proposition}

We pause to give some remarks on the origin of the above formula. The proof (given in Section~\ref{sec:pf-adv-bound}) follows a strategy based on~\cite{SW-recovery}: apply Jensen's inequality to $X$ (but not $Z$) and then the result can be explicitly calculated by solving an upper-triangular linear system. The original work~\cite{SW-recovery} gave a similar formula in the setting of \emph{estimation}, and more recently~\cite{planted-planted} was first to demonstrate that related techniques can also be used for testing between two ``planted'' distributions (which is also the setting of the current work). In contrast, previous low-degree lower bounds for testing problems had always required the ``null'' distribution $\QQ$ to have independent coordinates; see the remark below for comparison.

\begin{remark}\label{rem:gen}
We note that Proposition~\ref{prop:adv-bound} generalizes a well known formula for low-degree testing between ``signal'' and ``pure noise.'' Specifically, consider the case where $X = -\One$ so that $\QQ$ is i.i.d.\ Rademacher, and $\PP$ is any distribution on $\{\pm 1\}^N$. In this case $M_{\beta\alpha} = \One_{\beta=\alpha}$ and so Proposition~\ref{prop:adv-bound} reduces to the bound
\[ \Adv_{\le D}^2 \le \sum_{|\alpha| \le D} \left(\EE_{Y \sim \PP}[Y^\alpha]\right)^2, \quad Y^\alpha \coloneqq \prod_{i \in \alpha} Y_i,\]
which is standard (and in fact holds with equality); see Section~2.3 of~\cite{hopkins-thesis}.
\end{remark}

Returning to the proof, a more convenient parametrization for $w_\alpha$ will be $\hat{w}_\alpha = M_{\alpha\alpha} w_\alpha$. In this case, since $M_{\emptyset \alpha} = \E_\QQ[Y^\alpha]$, the recurrence can be written as
\[ \hat{w}_\emptyset = 1, \]
\begin{equation}\label{eq:w-hat-recur}
\hat{w}_\alpha \;=\; c_\alpha - \sum_{\beta \subsetneq \alpha} \hat{w}_\beta \frac{M_{\beta\alpha}}{M_{\beta\beta}} \;=\; \EE_\PP[Y^\alpha] - \EE_\QQ[Y^\alpha] - \sum_{\emptyset \subsetneq \beta \subsetneq \alpha} \hat{w}_\beta \frac{M_{\beta\alpha}}{M_{\beta\beta}} \qquad \text{for } |\alpha| \ge 1.
\end{equation}
The ratio of $M$'s can be thought of as a conditional probability:
\begin{equation}\label{eq:R}
R_{\beta\alpha} \coloneqq \frac{M_{\beta\alpha}}{M_{\beta\beta}} = \frac{\Pr_X(\alpha \setminus X = \beta)}{\Pr_X(\beta \cap X = \emptyset)} = \Pr_X(\alpha \setminus X = \beta \;|\; \beta \cap X = \emptyset).
\end{equation}

From this point onward, we specialize to our testing problem of interest: $\PP = \MC(n,q)$ versus $\QQ = \MC(n,q+\ell)$. As discussed above, our goal is to show $\Adv_{\le D} = 1+o(1)$ by bounding the formula in~\eqref{eq:adv-bound}. The ``1'' comes from the $\alpha = \emptyset$ term, and we need to show that the rest of the sum is $o(1)$.

The following property of $\hat w$ will be key to the analysis; it is used crucially in the proof of Lemma~\ref{lem:w-bound}. Note that we can think of $\alpha$ as a subset of edges of the complete graph on $n$ vertices, and in this sense we can talk about $\alpha$ being connected or having connected components.

\begin{lemma}\label{lem:multiplicative}
If $\alpha$ has connected components $\alpha_1,\ldots,\alpha_t$ then $\hat w_\alpha = \prod_{i=1}^t \hat w_{\alpha_i}$.
\end{lemma}

\begin{proof}
It suffices to prove the claim in the case where $\alpha$ is comprised of two non-empty disjoint edge sets $\alpha_1, \alpha_2$ with no vertices in common (i.e., each $\alpha_i$ is a union of connected components). Once we establish $\hat w_\alpha = \hat w_{\alpha_1} \hat w_{\alpha_2}$ in this case, the general statement follows by induction.

Note that due to independence across connected components, $c_\alpha = c_{\alpha_1} c_{\alpha_2}$. Any $\beta \subseteq \alpha$ can be uniquely decomposed as $\beta = \beta_1 \cup \beta_2$ with $\beta_1 \subseteq \alpha_1$ and $\beta_2 \subseteq \alpha_2$. Again by independence, $R_{\beta\alpha} = R_{\beta_1 \alpha_1} R_{\beta_2 \alpha_2}$. We will also need the fact $R_{\alpha\alpha} = 1$. We proceed by induction on $|\alpha|$. If either $\alpha_1$ or $\alpha_2$ is empty, the result follows immediately because $\hat{w}_\emptyset = 1$. Otherwise, assume by induction that $\hat w_\beta = \hat w_{\beta_1} \hat w_{\beta_2}$ for any $\beta \subsetneq \alpha$. We have
\begin{align*}
\hat w_\alpha &= c_\alpha - \sum_{\beta \subsetneq \alpha} \hat w_\beta R_{\beta\alpha} \\
&= c_{\alpha_1} c_{\alpha_2} - \sum_{\substack{\beta_1 \subsetneq \alpha_1 \\ \beta_2 \subsetneq \alpha_2}} \hat w_{\beta_1} \hat w_{\beta_2} R_{\beta_1 \alpha_1} R_{\beta_2 \alpha_2} - \sum_{\substack{\beta_1 \subsetneq \alpha_1 \\ (\beta_2 = \alpha_2)}} \hat w_{\beta_1} \hat w_{\alpha_2} R_{\beta_1 \alpha_1} R_{\alpha_2 \alpha_2} - \sum_{\substack{\beta_2 \subsetneq \alpha_2 \\ (\beta_1 = \alpha_1)}} \hat w_{\alpha_1} \hat w_{\beta_2} R_{\alpha_1 \alpha_1} R_{\beta_2 \alpha_2} \\
&= c_{\alpha_1} c_{\alpha_2} - \left(\sum_{\beta_1 \subsetneq \alpha_1} \hat w_{\beta_1} R_{\beta_1 \alpha_1}\right)\left(\sum_{\beta_2 \subsetneq \alpha_2} \hat w_{\beta_2} R_{\beta_2 \alpha_2}\right) - \hat w_{\alpha_2} \sum_{\beta_1 \subsetneq \alpha_1} \hat w_{\beta_1} R_{\beta_1 \alpha_1} - \hat w_{\alpha_1} \sum_{\beta_2 \subsetneq \alpha_2} \hat w_{\beta_2} R_{\beta_2 \alpha_2}.
\end{align*}
Using the recurrence~\eqref{eq:w-hat-recur}, this becomes
\[ \hat{w}_\alpha = c_{\alpha_1} c_{\alpha_2} - (c_{\alpha_1} - \hat{w}_{\alpha_1})(c_{\alpha_2} - \hat{w}_{\alpha_2}) - \hat w_{\alpha_2} (c_{\alpha_1} - \hat{w}_{\alpha_1}) - \hat w_{\alpha_1} (c_{\alpha_2} - \hat{w}_{\alpha_2}), \]
which simplifies to $\hat{w}_{\alpha_1} \hat{w}_{\alpha_2}$ as desired.
\end{proof}

\subsubsection{Bounding $\hat{w}_\alpha$}

In the remainder of the proof we need to bound the values $w_\alpha$ and plug this into~\eqref{eq:adv-bound}. Recall that when $\alpha$ is thought of as a graph, $|\alpha|$ is the number of edges. We also define $V(\alpha)$ to be the set of vertices of $\alpha$, i.e., the vertices $i \in [n]$ incident to at least one edge of $\alpha$.

\begin{lemma}\label{lem:M-diag}
For any $\alpha$ we have $M_{\alpha\alpha} \ge 1 - \frac{|\alpha|}{q+\ell}$.
\end{lemma}

\begin{proof}
Recall that $M_{\alpha\alpha}$ is the probability (under $\QQ$) that $\alpha$ contains no clique edges. The probability that any specific edge is a clique edge is $1/(q+\ell)$, so the result follows by a union bound.
\end{proof}

\begin{lemma}\label{lem:moment-diff}
If $|\alpha| \ge 1$ and $\alpha$ is connected then
\[ 0 \le \EE_\PP[Y^\alpha] - \EE_\QQ[Y^\alpha] \le \frac{\ell}{q^{|V(\alpha)|}} \, (|V(\alpha)|-1). \]
\end{lemma}

\begin{proof}
Since $\alpha$ is connected, $\E_\PP[Y^\alpha]$ is the probability that all vertices of $\alpha$ are assigned the same label in $[q]$ (and similarly for $\E_\QQ[Y^\alpha]$), i.e.,
\begin{align*}
\EE_\PP[Y^\alpha] - \EE_\QQ[Y^\alpha] &= \left(\frac{1}{q}\right)^{|V(\alpha)|-1} - \left(\frac{1}{q+\ell}\right)^{|V(\alpha)|-1} \\
&= \left(\frac{1}{q}\right)^{|V(\alpha)|-1}\left[1 - \left(\frac{q}{q+\ell}\right)^{|V(\alpha)|-1}\right] \\
&= \left(\frac{1}{q}\right)^{|V(\alpha)|-1}\left[1 - \left(1 - \frac{\ell}{q+\ell}\right)^{|V(\alpha)|-1}\right] \\
&\le \left(\frac{1}{q}\right)^{|V(\alpha)|-1}\left[1 - \left(1 - \frac{\ell}{q+\ell} \, (|V(\alpha)|-1)\right)\right] \\
&= \left(\frac{1}{q}\right)^{|V(\alpha)|-1} \frac{\ell}{q+\ell} \, (|V(\alpha)|-1) \\
&\le \frac{\ell}{q^{|V(\alpha)|}} \, (|V(\alpha)|-1).
\end{align*}
\end{proof}

\begin{lemma}\label{lem:w-bound}
If $|\alpha| \ge 1$ then
\[ |\hat{w}_\alpha| \le \left(\frac{\sqrt \ell}{q}\right)^{|V(\alpha)|} (|\alpha|+1)^{|\alpha|}. \]
\end{lemma}

\begin{proof}
Proceed by induction on $|\alpha|$. First consider the case where $\alpha$ is not connected. Write $\alpha$ as the union of two non-empty disjoint edge sets $\alpha_1, \alpha_2$ with no vertices in common. By Lemma~\ref{lem:multiplicative} and the induction hypothesis,
\begin{align*}
|\hat{w}_\alpha| &= |\hat{w}_{\alpha_1}| \cdot |\hat{w}_{\alpha_2}| \le \left(\frac{\sqrt \ell}{q}\right)^{|V(\alpha_1)|} (|\alpha_1|+1)^{|\alpha_1|} \cdot \left(\frac{\sqrt \ell}{q}\right)^{|V(\alpha_2)|} (|\alpha_2|+1)^{|\alpha_2|} \\
&\le \left(\frac{\sqrt \ell}{q}\right)^{|V(\alpha_1)|+|V(\alpha_2)|} (|\alpha_1|+|\alpha_2|+1)^{|\alpha_1|+|\alpha_2|} \\
&= \left(\frac{\sqrt \ell}{q}\right)^{|V(\alpha)|} (|\alpha|+1)^{|\alpha|}
\end{align*}
as desired.

Now consider the case where $\alpha$ is connected. Using~\eqref{eq:w-hat-recur} and Lemma~\ref{lem:moment-diff},
\[ |\hat{w}_\alpha| \le \left|\EE_\PP[Y^\alpha] - \EE_\QQ[Y^\alpha]\right| + \sum_{\emptyset \subsetneq \beta \subsetneq \alpha} |\hat{w}_\beta| \cdot |R_{\beta\alpha}|
\le \frac{\ell}{q^{|V(\alpha)|}} \, (|V(\alpha)|-1) + \sum_{\emptyset \subsetneq \beta \subsetneq \alpha} |\hat{w}_\beta| \cdot |R_{\beta\alpha}|. \]
Using the definition~\eqref{eq:R} and the connectivity of $\alpha$, we can deduce (for any $\emptyset \subsetneq \beta \subsetneq \alpha$)
\[ 0 \le R_{\beta\alpha} \le \left(\frac{1}{q+\ell}\right)^{|V(\alpha)| - |V(\beta)|} \le \left(\frac{1}{q}\right)^{|V(\alpha)| - |V(\beta)|}, \]
because once we condition on the labels in $V(\beta)$, each vertex in $V(\alpha) \setminus V(\beta)$ has at most one possible label that would allow the event $\alpha \setminus X = \beta$ to occur. (More formally, any vertex $i \in V(\alpha) \setminus V(\beta)$ is connected to some vertex $j \in V(\beta)$ by a path using edges from $\alpha \setminus \beta$. Since every edge on this path must be a clique edge in order for $\alpha \setminus X = \beta$ to occur, $i$ must have the same label as $j$.) Now using the above bounds and the induction hypothesis,
\begin{align*}
|\hat{w}_\alpha| &\le \frac{\ell}{q^{|V(\alpha)|}} \, (|V(\alpha)|-1) + \sum_{\emptyset \subsetneq \beta \subsetneq \alpha} \left(\frac{\sqrt\ell}{q}\right)^{|V(\beta)|}  (|\beta|+1)^{|\beta|} \cdot \left(\frac{1}{q}\right)^{|V(\alpha)| - |V(\beta)|} \\
&\le \left(\frac{\sqrt\ell}{q}\right)^{|V(\alpha)|} \left[|V(\alpha)|-1 + \sum_{\emptyset \subsetneq \beta \subsetneq \alpha} (|\beta|+1)^{|\beta|}\right] \qquad \text{since } |V(\alpha)| \ge 2 \text{ and } |V(\beta)| \le |V(\alpha)| \\
&= \left(\frac{\sqrt\ell}{q}\right)^{|V(\alpha)|} \left[|V(\alpha)|-1 + \sum_{m=1}^{|\alpha|-1} \binom{|\alpha|}{m} (m+1)^{m}\right] \\
&\le \left(\frac{\sqrt\ell}{q}\right)^{|V(\alpha)|} \left[|V(\alpha)|-1 + \sum_{m=1}^{|\alpha|-1} \binom{|\alpha|}{m} |\alpha|^{m}\right] \\
&= \left(\frac{\sqrt\ell}{q}\right)^{|V(\alpha)|} \left[|V(\alpha)|-1 + (|\alpha|+1)^{|\alpha|} - 1 - |\alpha|^{|\alpha|}\right] \qquad \text{by the Binomial theorem}\\
&\le \left(\frac{\sqrt\ell}{q}\right)^{|V(\alpha)|} (|\alpha|+1)^{|\alpha|},
\end{align*}
where the last step used $|V(\alpha)| \le 2|\alpha| \le |\alpha|^{|\alpha|} + 1$.
\end{proof}

\subsubsection{Putting it together}

The rest of the proof is similar to the low-degree analysis of planted clique; see Section~2.4 of~\cite{hopkins-thesis}.

\begin{proof}[Proof of Theorem~\ref{thm:testing-lower}]

For any $|\alpha| \le D$, we have from Lemma~\ref{lem:M-diag} that
\[ M_{\alpha\alpha} \ge 1 - \frac{|\alpha|}{q+\ell} \ge 1 - \frac{D}{q} = 1-o(1), \]
due to our assumptions on $q$ and $D$. Applying Proposition~\ref{prop:adv-bound},
\[ \Adv_{\le D}^2 \le \sum_{|\alpha| \le D} w_\alpha^2 = 1 + \sum_{1 \le |\alpha| \le D} \left(\frac{\hat w_\alpha}{M_{\alpha\alpha}}\right)^2 \le 1 + (1+o(1)) \sum_{1 \le |\alpha| \le D} \hat w_\alpha^2. \]
Since our goal (by Lemma~\ref{lem:adv-sep}) is to show $\Adv_{\le D} = 1+o(1)$, it remains to show
\[ \sum_{1 \le |\alpha| \le D} \hat w_\alpha^2 = o(1). \]
This follows from Proposition~\ref{prop:loglog} below, using the bound on $|\hat w_\alpha|$ from Lemma~\ref{lem:w-bound} together with the assumption $q^2 \ge \ell n^{1+\epsilon}$.
\end{proof}

\begin{lemma}\label{lem:count-graphs}
For integers $t \ge 2$ and $D \ge 1$, the number of graphs $\alpha \subseteq \binom{n}{2}$ such that $|\alpha| \le D$ and $|V(\alpha)| = t$, is at most $n^t \min\{2^{t^2}, t^{2D}\}$.
\end{lemma}

\begin{proof}
The number of ways to choose $t$ vertices is $\binom{n}{t} \le n^t$. Once the vertices are chosen, we can upper-bound the total number of graphs with $\le D$ edges in two different ways: $2^{\binom{t}{2}} \le 2^{t^2}$ or $\left(\binom{t}{2}+1\right)^D \le (t^2)^D$.
\end{proof}

\begin{proposition}\label{prop:loglog}
Suppose there exist fixed constants $\delta > 0$ and $C > 0$ such that for $\alpha \subseteq \binom{n}{2}$ with $1 \le |\alpha| \le D$, we have a quantity $\phi_\alpha$ bounded by $|\phi_\alpha| \le n^{-\frac{1}{2}(1+\delta) \cdot |V(\alpha)|} (|\alpha|+1)^{C \cdot |\alpha|}$. If $D = D_n$ satisfies $D = o(\log n/\log \log n)^2$ then
\[ \sum_{1 \le |\alpha| \le D} \phi_\alpha^2 = o(1) \]
as $n \to \infty$.
\end{proposition}

\begin{proof}

Using Lemma~\ref{lem:count-graphs} and the fact $|\alpha| \le \binom{|V(\alpha)|}{2} \le |V(\alpha)|^2$, 

\[ \sum_{1 \le |\alpha| \le D} \phi_\alpha^2 \le \sum_{2 \le t \le \sqrt{D}} n^t 2^{t^2} \cdot n^{-(1+\delta)t} (t^2+1)^{2C t^2} \;+\; \sum_{\sqrt{D} \le t \le 2D} n^t t^{2D} \cdot n^{-(1+\delta)t} (D+1)^{2CD}. \]

\noindent Consider the first sum on the right-hand side above. The initial term $t = 2$ is $O(n^2 \cdot n^{-2(1+\delta)}) = o(1)$, and the ratio between terms $t+1$ and $t$ is
\[ n^{-\delta} \cdot 2^{2t+1} \cdot ((t+1)^2+1)^{2C(2t+1)} \left(\frac{(t+1)^2+1}{t^2+1}\right)^{2Ct^2} \le t^{O(t)} n^{-\delta} \le \sqrt{D}^{O(\sqrt{D})} n^{-\delta} \le \frac{1}{2} \]
for sufficiently large $n$, using the assumption $D = o\left(\frac{\log n}{\log \log n}\right)^2$.
Now consider the second sum. The initial term $t = \left\lceil \sqrt{D} \right\rceil$ is at most
\[ n^{-\delta \sqrt{D}} (\sqrt{D}+1)^{2D} (D+1)^{2CD} \le n^{-\delta \sqrt{D}} (D+1)^{2(C+1)D} = o(1), \]
and the ratio between terms $t+1$ and $t$ is
\[ n^{-\delta} \cdot \left(\frac{t+1}{t}\right)^{2D} \le n^{-\delta} \left(1 + \frac{1}{\sqrt D}\right)^{2D} \le n^{-\delta} \cdot e^{O(\sqrt{D})} \le \frac{1}{2} \]
for sufficiently large $n$.
\end{proof}

\subsubsection{Proof of Proposition~\ref{prop:adv-bound}}
\label{sec:pf-adv-bound}

The proof is similar to the lower bound for planted clique in~\cite[Section~3.5]{SW-recovery}. We give the details here for convenience.

Any degree-$D$ polynomial $f: \{\pm 1\}^N \to \RR$ has a unique expansion $f(Y) = \sum_{\alpha \subseteq [N],\,|\alpha| \le D} \hat{f}_\alpha Y^\alpha$. Write
\[ \EE_\PP[f(Y)] = \sum_{|\alpha| \le D} \hat{f}_\alpha \EE_\PP[Y^\alpha] = \langle c,\hat f \rangle \]
where, recall, the vector $c = (c_\alpha)$ is defined by
\[ c_\alpha = \EE_\PP[Y^\alpha]. \]
By Jensen's inequality,
\[ \EE_\QQ[f(Y)^2] \ge \EE_Z \left(\EE_X f(X \vee Z)\right)^2 \eqqcolon \EE_Z g(Z)^2 = \|\hat{g}\|^2 \]
where
\begin{align*}
g(Z) &= \EE_X f(X \vee Z) \\
&= \sum_{|\alpha| \le D} \hat{f}_\alpha \EE_X (X \vee Z)^\alpha \\
&= \sum_{|\alpha| \le D} \hat{f}_\alpha \sum_{0 \subseteq \beta \subseteq \alpha} Z^\beta \Pr_X\{\alpha \setminus X = \beta\} \\
&= \sum_\beta Z^\beta \sum_{\alpha \supseteq \beta} \hat{f}_\alpha \Pr_X\{\alpha \setminus X = \beta\}.
\end{align*}
In other words, $\hat{g} = M \hat{f}$ where, recall, the matrix $M = (M_{\beta\alpha})$ is defined by
\[ M_{\beta\alpha} = \One_{\beta \subseteq \alpha} \Pr_X\{\alpha \setminus X = \beta\}. \]
Note that $M$ is upper triangular and (by assumption) has positive entries on the diagonal, so $M$ is invertible. We have now shown $\EE_\QQ[f]^2 \ge \|\hat g\|^2 = \|M\hat f\|^2$ and so
\[ \Adv_{\le D} = \sup_{f \in \RR[Y]_{\le D}} \frac{\EE_\PP[f]}{\sqrt{\EE_\QQ[f^2]}} \le \sup_{\hat f} \frac{\langle c,\hat f \rangle}{\|M \hat f\|} = \sup_{\hat g} \frac{c^\top M^{-1} \hat g}{\|\hat g\|}, \]
which has optimizer $\hat g = (c^\top M^{-1})^\top$, yielding
\[ \Adv_{\le D} \le \|c^\top M^{-1}\| \eqqcolon \|w\| \]
where $w$ is the solution to $w^\top M = c^\top$. Solving for $w$ using the upper-triangular structure of $M$ gives the recurrence
\begin{equation}\label{eq:w-recurrence}
w_\alpha = \frac{1}{M_{\alpha\alpha}} \left(c_\alpha - \sum_{\beta \subsetneq \alpha} w_\beta M_{\beta\alpha}\right),
\end{equation}
completing the proof.

\section{Refuting Colorability}\label{sec:ref-hardness}

\subsection{Upper Bound}

We restate the theorem for the reader's convenience.

\thmrefupper*

\begin{proof}
Let $A$ denote the $\{\pm 1\}$-valued adjacency matrix of the \emph{complement} graph, with $0$'s on the diagonal; if the graph is $q$-colorable then $A$ has value $1$ within each color class. For an integer $m \ge 1$ to be chosen later, consider the polynomial $f(X) = (n/q - 1)^{-2m}\, \Tr(A^{2m})$, which has degree $2m$ in the input variables $X \in \{\pm 1\}^{\binom{n}{2}}$.

First we let $X \in \mathcal{R}_q$ and aim to show $f(X) \ge 1$. Let $S \subseteq [n]$ be the largest color class, so $|S| \ge n/q$. Let $\One_S \in \{0,1\}^n$ denote the indicator vector for $S$. Letting $\lambda_{\max} = \lambda_1 \ge \lambda_2 \ge \cdots \ge \lambda_n$ denote the eigenvalues of $A$,
\[ \lambda_{\max} \ge \frac{\One_S^\top A \One_S}{\|\One_S\|^2} = \frac{|S|(|S|-1)}{|S|} = |S|-1 \ge \frac{n}{q}-1 \]
and
\[ \lambda_{\max}^{2m} \le \sum_{i=1}^n \lambda_i^{2m} = \Tr(A^{2m}). \]
Combining these yields $\Tr(A^{2m}) \ge (n/q - 1)^{2m}$ and so $f(X) \ge 1$.

It remains to show $\EE[f^2] = o(1)$ when $X \sim G(n,1/2)$. Let $Y$ be an $n \times n$ symmetric matrix where $\{Y_{ij} \,: i \le j\}$ are i.i.d.\ $\mathcal{N}(0,1)$. By direct expansion and comparison of Rademacher moments to Gaussian ones, $\EE[\Tr(A^{2m})^2] \le \EE[\Tr(Y^{2m})^2]$. Using $\|Y\|$ to denote the spectral norm of $Y$, the bound of~\cite[Lemma~2.2]{BvH} gives
\[ \EE[\Tr(Y^{2m})^2] \le \EE[n^2 \|Y\|^{4m}] \le n^2 (2\sqrt{n} + 2\sqrt{4m})^{4m}. \]
Putting it together,
\[ \EE[f^2] \le \left(\frac{n}{q}-1\right)^{-4m} n^2 (2\sqrt{n} + 2\sqrt{4m})^{4m} = n^2 \left(\frac{2q(\sqrt{n} + \sqrt{4m})}{n-q}\right)^{4m}, \]
which is $o(1)$ under the conditions of the theorem.
\end{proof}

\subsection{Lower Bound}
\label{sec:ref-lower}

We restate the theorem for the reader's convenience.

\thmreflower*

\noindent In light of Proposition~\ref{prop:strat-II}, our goal is to show $\Adv_{\le D}(\PP,\QQ) = 1+o(1)$ where $\QQ = G(n,1/2)$ (and $Y \sim \QQ$ is encoded by an element of $\{\pm 1\}^{\binom{n}{2}}$) and $\PP$ is the planted distribution defined in Definition~\ref{def:planting}. Our starting point is the well-known formula from Remark~\ref{rem:gen}:
\[ \Adv_{\le D}^2 = \sum_{|\alpha| \le D} \left(\EE_{Y \sim \PP}[Y^\alpha]\right)^2, \]
where $\alpha \subseteq \binom{n}{2}$. We identify $\alpha$ with the graph whose edge set is $\alpha$, and write $V(\alpha) \subseteq [n]$ for the vertex set, i.e., the vertices incident to at least one edge in $\alpha$. Our first step is to bound the coefficients $\lambda_\alpha \coloneqq \EE_{Y \sim \PP}[Y^\alpha]$.

\subsubsection{Bounding the coefficients}

\begin{lemma}[Bounding $\lambda_\al$]
\label{lem:lambda-bound}
For any graph $\al \subseteq \binom{n}{2}$ we have
\[ |\lambda_\alpha| \coloneqq \left|\EE_{Y \sim \PP}[Y^\alpha]\right| \le O(q^{-3/4})^{|V(\alpha)|} \]
where $O(\cdot)$ hides an absolute constant factor.
\end{lemma}

\begin{proof}
If $\alpha = \cup_i \alpha_i$ is the decomposition of $\alpha$ into connected components, we have $\lambda_\alpha = \prod_i \lambda_{\alpha_i}$ due to independence across components. It therefore suffices to prove the result in the case where $\alpha$ is connected.

Let $c: V(\alpha) \to [q] \times [q]$ denote the latent assignment of labels $(a,b)$ to vertices from the definition of $\PP$ (Definition~\ref{def:planting}). We have
\[ \lambda_\al = \EE_c \EE_{Y \sim \PP|c}[Y^{\al}] =
\sum_c \Pr[c] \cdot \E[Y^\al|c]. \]
Note that $\E[Y^\al|c] = 0$ unless every edge in $\alpha$ is either an independent set edge or clique edge in $c$, and in this case,
\[ \E[Y^\al|c] = (-1)^{\#\text{ ind-set edges}}. \]

As a result, one possible upper bound on $|\lambda_\alpha|$ is the probability over $c$ that every edge in $\alpha$ is either an ind-set edge or clique edge. We can bound this probability as follows. Recall we are assuming $\alpha$ is connected, and explore the vertices of $\alpha$ according to a breadth-first search. The first vertex's label is unconstrained. Each edge that leads to a new vertex must be an ind-set edge or clique edge, giving at most $2q$ possibilities for the new vertex's label. Since there are $q^2$ possible labels in total, we conclude
\begin{equation}\label{eq:generic-lambda}
|\lambda_\alpha| \le \left(\frac{2q}{q^2}\right)^{|V(\alpha)|-1} = \left(\frac{2}{q}\right)^{|V(\alpha)|-1}
\end{equation}
for any connected $\alpha$.

The bound~\eqref{eq:generic-lambda} implies the desired result $|\lambda_\alpha| \le O(q^{-3/4})^{|V(\alpha)|}$ provided $|V(\alpha)| \ge 4$, as in this case we have $|V(\alpha)|-1 \ge |V(\alpha)| - \frac{1}{4}|V(\alpha)| = \frac{3}{4}|V(\alpha)|$. For $|V(\alpha)| \le 3$ we will manually verify the result by checking all the possible graphs:
\begin{itemize}
    \item If $\alpha$ has no edges then $\lambda_\alpha = 1$.
    \item If $\alpha$ is a single edge, the cases to consider for $c$ are $\{(a,b),(a,b)\}$, $\{(a,b),(a,b')\}$, and $\{(a,b),(a',b)\}$ (where $a \ne a'$, $b \ne b'$). This gives
    \[ \lambda_\alpha = -\frac{1}{q^2} -\frac{1}{q}\left(1-\frac{1}{q}\right) + \frac{1}{q}\left(1-\frac{1}{q}\right) = -q^{-2}. \]
    \item If $\alpha$ is a length-2 path then conditioned on any label for the middle vertex, the two edges are independent. Reusing the calculation for the single edge, we have $\lambda_\alpha = (-q^{-2})^2 = q^{-4}$.
    \item If $\alpha$ is a triangle, we first claim that the only labelings $c$ that contribute to $\lambda_\alpha$ are those in which a label $(a,b)$ is repeated. This follows from the symmetry between $c$ and the reversed labeling $\overline{c}$ where each pair is reversed: $(a,b) \mapsto (b,a)$. If $c$ has no repeated labels, $c$ and $\overline{c}$ contribute the same term but with opposite signs, as every ind-set edge becomes a clique edge and vice versa. In light of this, the remaining cases to consider for $c$ are $\{(a,b),(a,b),(a,b)\}$, $\{(a,b),(a,b),(a,b')\}$, and $\{(a,b),(a,b),(a',b)\}$. This gives
    \[ \lambda_\alpha = -\frac{1}{q^4} - 3 \cdot \frac{1}{q^3}\left(1-\frac{1}{q}\right) - 3 \cdot \frac{1}{q^3}\left(1-\frac{1}{q}\right) = O(q^{-3}). \]
\end{itemize}
We have now verified $|\lambda_\alpha| \le O(q^{-3/4})^{|V(\alpha)|}$ for every connected $\alpha$. As discussed previously, this implies the result for all $\alpha$.
\end{proof}

\subsubsection{Putting it together}

We now combine the results from above in order to bound $\Adv_{\le D}$.

\begin{proof}[Proof of Theorem~\ref{thm:ref-lower}]
Due to our assumption $q \ge n^{2/3+\epsilon}$, Lemma~\ref{lem:lambda-bound} gives
\[ |\lambda_\alpha| \le O(n^{-\frac{3}{4}(\frac{2}{3}+\epsilon)})^{|V(\alpha)|} = O(n^{-\frac{1}{2} -\frac{3}{4}\epsilon})^{|V(\alpha)|} \le n^{-\frac{1}{2}(1+\epsilon) \cdot |V(\alpha)|} \]
for sufficiently large $n$. Using Proposition~\ref{prop:loglog}, we have for any $D = o(\log n/\log \log n)^2$,
\[ \Adv^2_{\le D} - 1 = \sum_{1 \le |\alpha| \le D} \lambda_\alpha^2 = o(1). \]
As discussed at the beginning of Section~\ref{sec:ref-lower}, this completes the proof.
\end{proof}

\section{Completeness of Quiet Planting}

In this section, we give a simple argument showing that the absence of a computationally quiet planted distribution implies the existence of a low-degree refutation algorithm, in high generality. Our proof is elementary and only needs a simple application of von Neumann's min-max principle. We restate the theorem for the reader's convenience.

\thmrefduality*

\begin{proof}
Let $\mathcal{P}$ denote the space of probability distributions on $\mathcal{R}$. Let $\mathcal{F}$ denote the space of degree-$D$ polynomials $f: \RR^N \to \RR$ such that $\EE_\QQ[f] = 0$ and $\EE_\QQ[f^2] \le 1$. Consider
\begin{equation}\label{eq:val1}
\mathrm{val}_n = \inf_{\PP \in \mathcal{P}} \sup_{f \in \mathcal{F}} \EE_\PP[f].
\end{equation}
By von Neumann's min-max principle (see below for discussion of the technical conditions required), the supremum and infimum can be exchanged:
\begin{equation}\label{eq:val2}
\mathrm{val}_n = \sup_{f \in \mathcal{F}} \inf_{\PP \in \mathcal{P}} \EE_\PP[f] = \sup_{f \in \mathcal{F}} \inf_{X \in \mathcal{R}} f(X).
\end{equation}
A degree-$D$ polynomial strongly (respectively, weakly) separates $\QQ$ and $\mathcal{R}$ if and only if the value of~\eqref{eq:val2} is $\omega(1)$ (resp., $\Omega(1)$). The negation of this statement is that $\mathrm{val}_n = O(1)$ (resp., $o(1)$) for an infinite subsequence of $n$, which from~\eqref{eq:val1} is equivalent to having $\PP_n$ defined on an infinite subsequence such that $\sup_{f \in \mathcal{F}} \EE_\PP[f] = O(1)$ (resp., $o(1)$). Now the result follows due to the identity $(\sup_{f \in \mathcal{F}} \EE_\PP[f])^2 + 1 = \Adv_{\le D}^2(\PP,\QQ)$; see Lemma~\ref{lem:adv-mean} below.

It remains to verify the technical conditions for the min-max principle. Formally we use the following variant, which is a special case of Sion's min-max theorem~\cite{sion-1,sion-2}.

\begin{theorem}
Let $\mathcal{P}$ be a compact convex subset of a linear topological space and $\mathcal{F}$ a convex subset of a linear topological space. If $\phi(x,y)$ is a continuous real-valued function on $\mathcal{P} \times \mathcal{F}$ with $\phi(x,\cdot)$ concave for all $x \in \mathcal{P}$, and $\phi(\cdot,y)$ convex for all $y \in \mathcal{F}$, then $\min_{x \in \mathcal{P}} \sup_{y \in \mathcal{F}} \phi(x,y) = \sup_{y \in \mathcal{F}} \min_{x \in \mathcal{P}} \phi(x,y)$.
\end{theorem}

\noindent In our setting, the linear topological spaces will simply be $\RR^d$ for some $d$. Recall that our choice of $\mathcal{P}$ is the space of probability distributions on a finite set $\mathcal{R} = \{r_1,r_2,\ldots,r_{|\mathcal{R}|}\}$. We can identify $\mathcal{P}$ with a compact convex subset of $\RR^{|\mathcal{R}|}$ by encoding a distribution $\PP$ as the vector of probabilities $(\PP(r_1),\ldots,\PP(r_{|\mathcal{R}|}))$. Recall that our choice of $\mathcal{F}$ is the space of degree-$D$ polynomials $f: \RR^N \to \RR$ such that $\EE_\QQ[f] = 0$ and $\EE_\QQ[f^2] \le 1$. Letting $\mathcal{X} = \supp(\QQ) \cup \mathcal{R} = \{x_1,\ldots,x_{|\mathcal{X}|}\}$, we can identify $\mathcal{F}$ with a convex subset of $\RR^{|\mathcal{X}|}$ (note that $\mathcal{F}$ is not required to be compact) by encoding a function $f: \RR^N \to \RR$ as the vector $(f(x_1),\ldots,f(x_{|\mathcal{X}|}))$. Finally, note that $\phi(\PP,f) \coloneqq \EE_\PP[f]$ is continuous, convex in $\PP$, and concave in $f$; in fact, it is linear in both variables. This justifies our earlier exchange of inf and sup, completing the proof.
\end{proof}

\begin{remark}\label{rem:compact}
Above we have assumed $\supp(\QQ)$ and $\mathcal{R}$ are finite to simplify the analytic conditions needed for the min-max principle, but these assumptions can be relaxed. For instance, one can alternatively assume that $\QQ_n$ is any distribution on $\RR^N$ with all moments finite and that $\mathcal{R}_n \subseteq \RR^N$ is compact. Since $\mathcal{R}$ is compact, the space $\mathcal{P}$ of probability distributions on $\mathcal{R}$ is compact in the weak-* topology.
\end{remark}

\begin{lemma}\label{lem:adv-mean}
$\sup_{f \in \mathcal{F}} \EE_\PP[f]^2 + 1 = \Adv_{\le D}^2(\PP,\QQ)$.
\end{lemma}

\begin{proof}
If the likelihood ratio $LR = d\PP/d\QQ$ exists, this fact follows from standard characterizations of these quantities as $L^2(\QQ)$-norms of projections of likelihoods (see Section~2.3 of~\cite{hopkins-thesis}); namely, the left-hand side is $\|LR^{\le D}-1\|_\QQ^2+1$ and the right-hand side is $\|LR^{\le D}\|_\QQ^2$. We also give a self-contained proof below.

Recalling the definition of $\Adv_{\le D}$, our goal is to show
\[ \sup_{f \in \mathcal{F}} \EE_\PP[f]^2 + 1 = \sup_{g \in \RR[Y]_{\le D}} \frac{\EE_\PP[g]^2}{\EE_\QQ[g^2]}. \]
Note that the value 1 is achievable on both sides by taking $f = 0$ or $g = 1$. To show ``$\le$,'' suppose we have $f \in \mathcal{F}$ such that $\EE_\PP[f] = a > 0$, achieving value $a^2+1$ on the left-hand side. Then $g = f + 1/a$ achieves the same value $a^2+1$ on the right-hand side.

To show ``$\ge$,'' suppose $g$ achieves value $b^2 > 1$ on the right-hand side, and scale $g$ so that $\EE_\QQ[g^2] = 1$ and $\EE_\PP[g] = b > 1$. Define $\Delta = \EE_\QQ[g]$ and note that $\EE_\QQ(g-\Delta)^2 = 1 - \Delta^2 \ge 0$ and $\EE_\PP(g-\Delta) = b - \Delta > 0$. If $\Delta = 1$ then the left-hand side is unbounded by taking $f$ to be an arbitrary multiple of $g-\Delta$. Otherwise set $f = (g-\Delta)/\sqrt{1-\Delta^2} \in \mathcal{F}$ and compute the left-hand side value
\[ \EE_\PP[f]^2 + 1 = \frac{(b-\Delta)^2}{1-\Delta^2} + 1 = b^2 + \frac{(b\Delta-1)^2}{1-\Delta^2} \ge b^2, \]
completing the proof.
\end{proof}

\appendix

\section{Planted $(q+1)$-coloring is not $q$-colorable}
\label{app:coloring}

Here we work with the complement graph and consider a partition into cliques rather than a coloring. Recall the multiple cliques model (Definition~\ref{def:MC}).

\begin{proposition}
If $1 \le q \le \Omega(n/\log n)$ then with probability $1-o(1)$, $\MC(n,q+1)$ does not admit a partition of the vertices into $q$ cliques.
\end{proposition}

\begin{proof}
Fix an absolute constant $\epsilon > 0$, to be chosen later. Assume $q \le c n /\log n$ for a constant $c = c(\epsilon) > 0$ to be chosen later. The proof hinges on 3 basic facts, which hold w.h.p.:
\begin{itemize}
    \item[(i)] $G(n,1/2)$ does not contain the complete bipartite graph $K_{m,m}$ as a subgraph, for $m \ge (2+\epsilon) \log_2 n$.
    \item[(ii)] Letting $S_1,\ldots,S_{q+1}$ denote the color classes of $\MC(n,q+1)$, we have $|S_i| \in (1 \pm \epsilon) \frac{n}{q+1}$ for all $i \in [q+1]$.
    \item[(iii)] In $\MC(n,q+1)$, any vertex $v \in S_i$ has at most $(1/2 + 2\epsilon)\frac{n}{q+1}$ neighbors in $S_j$, for $i \ne j$.
\end{itemize}
Standard arguments show that (i)--(iii) hold with probability $1-o(1)$, and we omit the details. The proof of (i) is a first moment calculation (compute the expected number of copies of $K_{m,m}$ and apply Markov's inequality), and the proof of (ii) and (iii) uses Bernstein's inequality along with a union bound.

Suppose $G \sim \MC(n,q+1)$. To complete the proof, it suffices to show that properties (i)--(iii) deterministically imply that $G$ has no partition into $q$ cliques (where property (i) applies to the underlying random graph $G' \sim G(n,1/2)$ used to generate $G$, before the $q+1$ cliques were added). Assume (i)--(iii) hold, and suppose for contradiction that $G$ admits a partition $V(G) = T_1 \sqcup T_2 \sqcup \cdots \sqcup T_q$ into cliques.

We first claim that for every $i \in [q]$, we either have (Case I) $|T_i| \le \frac{3}{4} \cdot \frac{n}{q+1}$ or (Case II) for some $j \in [q+1]$, $T_i \subseteq S_j$ and $|T_i| > \frac{1}{2}|S_j|$. To see this, note that if Case I fails then $|T_i| > \frac{3}{4} \cdot \frac{n}{q+1}$, and so to avoid violating property (i) there must exist $j$ such that $|T_i \cap S_j| \ge \frac{3}{4} \cdot \frac{n}{q+1} - (2+\epsilon) \log_2 n > (1/2 + 2\epsilon) \frac{n}{q+1}$. Now property (iii) implies $T_i \subseteq S_j$. Also, property (ii) implies $|T_i| > \frac{1}{2}|S_j|$. This proves the claim.

Using the above claim, we can now construct an injective map $\phi: [q] \to [q+1]$ such that $|T_i| \le |S_{\phi(i)}|$. First, for $i$ in Case II, set $\phi(i)$ to be the corresponding $j$; then for $i$ in Case I, set $\phi(i)$ to be any unused $j$ value. Since some $j \in [q+1]$ is not in the image of $\phi$, we have $\sum_{i \in [q]} |T_i| < \sum_{j \in [q+1]} |S_j| = n$, a contradiction.
\end{proof}

\bibliographystyle{alpha}
\bibliography{main}

\end{document}